\documentclass[11pt]{article}

\usepackage[left=1in, top=1in, right=1in, bottom=1in]{geometry}

\usepackage[numbers]{natbib}
\usepackage[utf8]{inputenc} 
\usepackage{nicefrac}       
\usepackage{amsmath,amsthm,amsfonts,dsfont}
\usepackage{amssymb}
\usepackage{comment}
\usepackage{graphicx}
\usepackage{booktabs} 
\usepackage[ruled]{algorithm2e} 
\usepackage{algorithmic}

\usepackage{enumitem}
\usepackage{wrapfig}
\usepackage{xcolor}
\usepackage{thm-restate}
\usepackage{hyperref}
\usepackage[capitalize]{cleveref}

\usepackage[suppress]{color-edits}
\addauthor{ag}{blue}
\addauthor{vg}{red}
\addauthor{sb}{magenta}
\addauthor{bj}{olive}

\usepackage{bbm}

\SetAlFnt{\small}
\SetAlCapFnt{\small}
\SetAlCapNameFnt{\small}
\SetAlCapHSkip{0pt}
\IncMargin{-\parindent}

\def\algoname{Set-Aside Greedy}
\def\totalvalue{1}

\newcommand{\pv}{\mathbf{p}}
\newcommand{\bz}{\mathbf{z}}
\newcommand{\bx}{\mathbf{x}}
\newcommand{\xv}{\mathbf{x}}
\newcommand{\bv}{\mathbf{v}}

\newcommand{\CRnsw}{\gamma^{\mathsf{NSW}}}
\newcommand\wtd[1]{\widetilde{#1}}
\usepackage{multicol}

\newcommand\abs[1]{\left\lvert{#1}\right\rvert}

\newtheorem{theorem}{Theorem}
\newtheorem{proposition}{Proposition}

\newtheorem{lemma}{Lemma}

\newtheorem{definition}{Definition}

\newtheorem{example}{Example}

\newcommand*\samethanks[1][\value{footnote}]{\footnotemark[#1]}

\usepackage{hyperref}       
\usepackage{url}            

\title{Online Nash Social Welfare Maximization with Predictions}

\author{Siddhartha Banerjee\thanks{School of ORIE, Cornell University, Ithaca, NY. Supported by NSF grants ECCS-1847393, DMS-1839346, CNS-1955997, CCF-1908517, and NSERC fellowship PGSD3-532673-2019. Emails: \texttt{\{sbanerjee, bzj3\}@cornell.edu}.}  \and Vasilis Gkatzelis\thanks{Department of Computer Science, Drexel University. Supported by NSF grant CCF-1755955. Email: \texttt{gkatz@drexel.edu}} \and Artur Gorokh\thanks{Facebook, New York, NY. Email: \texttt{ag2282@cornell.edu}.} \and Billy Jin\samethanks[1]
}

\date{}

\begin{document}

\maketitle


\begin{abstract}
We consider the problem of allocating a set of divisible goods to $N$ agents in an online manner, aiming to maximize the Nash social welfare, a widely studied objective which provides a balance between fairness and efficiency. The goods arrive in a sequence of $T$ periods and the value of each agent for a good is adversarially chosen when the good arrives.
We first observe that no online algorithm can achieve a competitive ratio better than the trivial $O(N)$, unless it is given additional information about the agents' values. 

Then, in line with the emerging area of ``algorithms with predictions'', we consider a setting where for each agent, the online algorithm is only given a prediction of her \emph{monopolist utility, i.e., her utility if all goods were given to her alone (corresponding to the sum of her values over the $T$ periods)}. Our main result is an online algorithm whose competitive ratio is parameterized by the multiplicative errors in these predictions. The algorithm achieves a competitive ratio of $O(\log N)$ and $O(\log T)$ if the predictions are perfectly accurate. Moreover, the competitive ratio degrades smoothly with the errors in the predictions, and is surprisingly robust: the logarithmic competitive ratio holds even if the predictions are very inaccurate.

We complement this positive result by showing that our bounds are essentially tight: no online algorithm, even if provided with perfectly accurate predictions, can achieve a competitive ratio of 
$O(\log^{1-\epsilon} N)$ or $O(\log^{1-\epsilon} T)$ for any constant $\epsilon>0$.


\end{abstract}





\pagenumbering{arabic}
\setcounter{page}{1}

\section{Introduction}
\label{sec:introduction}

We study an online resource allocation problem where each day some divisible resource becomes available, and we need to design an online algorithm that distributes it among a set of $N$ agents, aiming to reach an outcome that combines fairness and efficiency. At the start of each day $t$, every agent $i$ informs the algorithm about their value $v_{i,t}$ for that day's resource, and the algorithm irrevocably decides how to split this resource, without knowing the agents' future values. If agent~$i$ is allocated a fraction $x_{i,t}$ of the resource, then her utility increases by $x_{i,t}v_{i,t}$, and after all the resources have been allocated, the total utility of an agent $i$ is $u_i(\bx) =  \sum_t v_{i,t} x_{i,t}$. 

If one were only concerned about efficiency, a common objective to maximize is the utilitarian social welfare -- that is, the sum of the agents' utilities, $\sum_i u_i(\bx)$. This is easy to maximize even in an online setting: allocate the entire resource to the agent $i$ with the highest value $v_{i,t}$ on each day~$t$. However, it is easy to see that this approach can ``starve'' many of the agents by never allocating anything to them, which is unacceptable in many settings of interest.

As an example, consider a setting where the processing time of a shared computing cluster needs to be divided among the employees of a firm or university. Each user may have a high value for gaining access to the cluster on some days (e.g., due to a conference deadline), and be willing to pass up on her access on other days. Although it makes sense to prioritize the users who need the resource the most each day, every user deserves some access to this shared resource.
For another example that has received some attention, consider a food bank that allocates food each day to soup kitchens and other local charities (see e.g. \citep{prendergast2017}). The number of people coming to each distribution facility varies from day to day, affecting each facility's demand for food, but, apart from efficiency, the food bank also wants to ensure some measure of equitability in its allocation.

Motivated by the need for fairness considerations when allocating these resources, we turn to the literature on fair division. An extreme measure of fairness from this literature is the egalitarian social welfare, or \emph{maximin welfare} (MW), which is equal to the minimum utility over all agents, $\min_i \{u_i(\bx)\}$. This objective interprets fairness as making sure that the least satisfied agent is as happy as possible. However, maximizing this objective can lead to highly inefficient outcomes by allocating resources to agents that are hard to satisfy instead of agents that can really benefit from them; moreover, the allocation is highly sensitive to perturbations in valuations. Although we briefly address this objective, the focus of this paper is on the \emph{Nash social welfare} (NSW), a much less extreme objective that is known to provide a natural balance between fairness and efficiency. The NSW objective is equal to the geometric mean of the agents' utilities, $\prod_i u_i(\bx)^{1/N}$; it was initially proposed about 70 years ago~\citep{Nash50,KN79}, and has recently received a great deal of attention. Apart from a balance between fairness and efficiency, it is also known to satisfy other desirable properties, including scale-independence, meaning that the scale of the agents' valuations does not affect the NSW maximizing solution~\citep{Moulin03}. 

Unfortunately, in contrast to utilitarian social welfare which is easy to optimize in the online setting, the problem becomes much more challenging when fairness considerations are introduced. Although computing an allocation that maximizes the NSW (or the MW) objective is computationally tractable problem in an offline  setting~\cite[Chapter 5]{AGT}, we prove that even sublinear competitive ratios are impossible to achieve online. To get some intuition for this, note that achieving fairness may require allocating resources to agents other than the one with the highest value for them. To decide if and when this should happen, offline algorithms use agents' valuations for \emph{all} the resources, whereas online algorithms have only historical information and cannot foresee which agents will be hard to satisfy in the future. 
In line with the exciting emerging literature on ``algorithms with predictions'' (e.g., see~\cite{AlgorithmsWithPredictions}), our goal in this paper is to overcome this overly pessimistic impossibility result by analyzing the performance of online algorithms that are augmented with predictions.

\subsection{Our Results}

We first observe that, in the absence of any predictions regarding the agents' values, no online algorithm aiming to maximize the Nash social welfare can achieve a competitive ratio better than $O(\min\{N,T\})$, where $N$ is the number of agents and $T$ is the total number of days (or resources). However, this pessimistic result heavily depends on the (often very unrealistic) assumption that the algorithm has \emph{no} information regarding agents' values.

To overcome analogous impossibility results, recent work on \emph{algorithms with predictions} or \emph{learning-augmented algorithms} instead assumes that the algorithms are equipped with some side-information, e.g., learned from historical data. Following this approach, our main result is an online algorithm, the \algoname\ Algorithm, which is augmented with a prediction $\wtd{V}_i$ of \sbedit{each agent $i$'s total \emph{monopolist} value $V_i=\sum_t v_{i,t}$, i.e., her utility if she was allocated \emph{all} the items. Note that this is a fairly mild amount of side-information -- only one number per agent. Nevertheless,} \algoname\ achieves an exponential improvement over the $O(\min\{N,T\})$ hardness result, achieving a competitive ratio of $O(\log(\min\{N,T\}))$ if the predictions are perfect. Moreover, the competitive ratio of \algoname\ \sbreplace{can be parameterized by the multiplicative errors in the predictions}{is good even with very poor predictions}. Specifically, suppose each agent $i$ overestimates her true total value by at most a factor $c_i$, or underestimates it by at most a factor $d_i$, i.e., $\wtd{V_i}\in \left[\frac{1}{d_i} V_i, ~ c_i V_i\right]$. The main result of this paper is the following theorem, which parameterizes the competitive ratio of \algoname\ as a function of $c_i$ and $d_i$.
\begin{restatable}{theorem}{main}
\label{thm:main_result}
The \algoname\ algorithm achieves competitive ratio 
$$
\CRnsw
\leq 
\left(\prod_{i=1}^N c_i\right)^{\frac{1}{N}}  
\min\left\{\log(2N) + \frac{1}{N}\sum_{i=1}^N\log(d_i), \, \log(2T) + \log(\max_i\{d_i\}) \right\}.$$
\end{restatable}
The competitive ratio is very robust: $O(\log(\min\{N,T\}))$ holds as long as $\left(\prod_{i=1}^N c_i\right)^{\frac{1}{N}}=O(1)$ and $\frac{1}{N}\sum_{i=1}^N\log(d_i)=O(\log(\min\{N,  T\}))$ (e.g., it holds if all the agents are underestimated by a factor $\min\{N, T\}$ or a $1/\log(N)$ fraction of the agents' values are overestimated by a factor $N$) 

The algorithm works by dividing each resource in two halves and allocating each half via a different approach: the first half is split equally between agents and the second half is allocated in a greedy manner aiming to myopically maximize the NSW. However, this combination of the two approaches would fail, were it not for a crucial adjustment to the agent utilities used in the greedy portion of the algorithm. In particular, on each day $t$, instead of using the actual utility that each agent has accrued up to that day, the greedy algorithm uses a ``predicted utility", which for each agent on day $t$ corresponds to the actual utility she has accrued from the greedy allocation up to that day, plus a $1/(2N)$ fraction of her predicted total utility $\wtd{V}_i$. Crucially, though the agent has not yet accrued this additional utility, she is \emph{guaranteed} to eventually receive some fraction of it (based on the quality of the prediction) from her promised share of the first half of the resources.


Our second main result shows that the \algoname\ algorithm is essentially optimal, by proving that no online algorithm, even if it knows the exact $V_i$ values for every $i$, can achieve a competitive ratio of $O(\log^{1-\epsilon}{N})$ or $O(\log^{1-\epsilon}{T})$ for any constant $\epsilon>0$ (see \cref{thm:main_negative_result}). This construction is quite involved so, in order to provide some intuition, we first prove an analogous result for the maximin welfare objective. Specifically, we show that no online algorithm can achieve a competitive ratio better than $O(\sqrt{N})$ for the maximin objective, even if it is equipped with perfect predictions $\wtd{V}_i=V_i$ (see \cref{thm:impossibility_maxmin}). Although this objective is quite different than the Nash social welfare, the technique we use for this result resembles the (much more demanding) one used for the Nash social welfare, and the proof also provides some intuition regarding the types of obstacles that arise when introducing fairness considerations in online resource allocation.

\medskip
\noindent\textbf{Technical Highlights.}
The use of predicted utilities plays a central role in our algorithm. From an algorithmic standpoint, it avoids a ``cold start'' problem that the myopic greedy approach suffers from. For example, suppose there are agents whose accrued utility so far is low (or zero). In this case, a myopic greedy approach would overcompensate by allocating large portions of the resources to them, even if their values for these resources are low. In~\cref{prop:myopic_greedy}, we make this intuition precise by showing an instance for which myopic greedy gets a competitive ratio of $\Omega(N)$.

In contrast, \algoname\ circumvents this issue by first predicting a minimum utility of $\wtd{V}_i/(2N)$ for each agent $i$ (using the first half of each resource), and then allocating the second half in a greedy manner \emph{accounting for the $\wtd{V}_i/(2N)$ future utility predicted for each agent $i$}. This provides a novel way for an online algorithm to use a prediction and could be of independent interest within the literature on algorithms with predictions when dealing with multiagent resource allocation problems.


Apart from their conceptual significance, our use of predictions in competitive analysis may also be of independent interest. 
At a high level, we leverage the Eisenberg-Gale program for optimal NSW allocations to construct a dual certificate to bound the optimality gap of an arbitrary allocation. We then use a novel potential function to track changes in this dual certificate. 
While our approach is in the spirit of regret analysis in online learning, it introduces several new ideas to deal with the NSW objective using predictions, which may prove of use in other settings.

\medskip
\noindent\textbf{Paper Structure.}
We formalize our setting in \cref{sec:setting}, and in \cref{sec:challenges} provide some intuition regarding the algorithmic challenges that arise when introducing fairness into online settings. 
\cref{sec:positiveresults} contains the main positive result of the paper: the \algoname\ algorithm, that achieves a competitive ratio of $O(\log N)$ and $O(\log T)$ for online NSW maximization. 
\cref{sec:lowerbnds} complements this with hardness results, showing that even with perfect knowledge of all agents' monopolist values, no online algorithm can achieve a competitive ratio of better than $O(\sqrt{N})$ for MW (\cref{ssec:lowerbndminmax}), and $O(\log^{1-\epsilon} N)$ or $O(\log^{1-\epsilon} T)$ for any $\epsilon > 0$ for NSW (\cref{sec:hardnessNSW}); the former highlights the lack of robustness of MW, while the latter matches our positive result for NSW.

\subsection{Related Work}
\label{ssec:relatedwork}



The online algorithm that we propose in this paper is a contribution to the exciting emerging literature that moves beyond worst-case analysis by focusing on \emph{algorithms with predictions} or \emph{learning-augmented algorithms} (e.g., see the book chapter by \citet{AlgorithmsWithPredictions}). In contrast to the overly pessimistic model of worst-case analysis, this line of work instead assumes that algorithms are enhanced with some exogenously provided prediction regarding some of the relevant parameters of the problem. The quality of the algorithm is then evaluated based on the quality of the prediction, measuring how the competitive ratio deteriorates as a function of the error in the prediction. Our work contributes to this literature by proposing a natural type of parameter to predict in multiagent resource allocation problems, as well as a novel way to leverage this prediction in an algorithm to achieve an exponential improvement in the competitive ratio.

Our work builds on a long literature on using primal-dual style algorithms for online resource allocation. One relevant example is the work of \citet{DJ2012}, which considers allocating items to agents in an online setting to maximize a generic objective $F(x) = \sum_i M(\sum_t x_{it}v_{it})$, where $M$ is a non-decreasing concave function. While the offline Nash social welfare maximizer can be computed by picking $F(x) = \sum_i \log (\sum_t x_{it}v_{it})$  (the so-called Eisenberg-Gale program), the multiplicative approximation guarantees achieved via this approach do not translate to meaningful guarantees with respect to the geometric mean. 
This is closer in spirit to the work of~\citet{azar2010allocate}, who consider the same problem, but with the objective of balancing the fraction of their total utility that each agent gets in hindsight. Compared to our bounds, the guarantees they obtain  depend on the conditioning of the valuations (specifically, on the ratio of the maximum value of an agent to the minimum nonzero value). Hence, their bound on the competitive ratio can be arbitrarily large as this ratio goes to infinity. These works emphasize how critical our use of predictions 
is in getting strong guarantees for online NSW maximization.


The main obstacles that we face in this paper are more closely related to the obstacles that commonly appear in the growing literature on online fair division. Specifically, defining and achieving fairness often requires a more holistic view of the instance at hand, which is exactly what online algorithms lack. For example, a common barrier in this work is the algorithm's inability to distinguish between agents that will be easy to satisfy later on, and those that will be hard to satisfy. 
\citet{GPT21} study the extent to which online algorithms can maximize the utilitarian social welfare, while satisfying envy-freeness. They make the simplifying assumption that the sum of every agent's valuations are normalized to 1, i.e., that $V_i=1$ for all $i$. This is a significantly more stringent assumption than the one that we make in this paper, where $V_i$ can take any value and we require only a rough estimate to achieve our bounds. However, even with this assumption, they show that the best possible approximation for instances with $n$ agents is $O(\sqrt{n})$, so they focus on designing mechanisms for instances with just two agents. 

\citet{BKPP18} consider a setting similar to ours, but the items being allocated in their case are indivisible, i.e., can be allocated only to a single agent (so they use randomness). Their goal is to minimize the amount of envy among the agents, and they show that the best way to minimize the expected maximum envy (up to sub-logarithmic factors) is to totally disregard the agent valuations and allocate each item uniformly at random among the agents, leading to a bound of $\tilde O(\sqrt{T/N})$. 
%
\citet{ZP20} revisit this setting and study the extent to which approximate envy-freeness can be combined with approximate Pareto efficiency. They consider a spectrum of increasingly powerful adversary models and they show that even for a non-adaptive adversary (which is weaker than the adaptive adversary model we consider in this work) there is no algorithm that can guarantee the aforementioned approximate envy while Pareto dominating the random allocation algorithm. They define an outcome to be $\alpha$ Pareto-efficient if improving every agent's utility by a factor $\alpha$ is infeasible, and they show that it is impossible to combine the envy bound with a $1/N$ approximation of Pareto efficiency, which is a trivial approximation that the random allocation algorithm satisfies. Note that our competitive ratio guarantees in this paper directly translate into approximate Pareto efficiency bounds. Specifically a competitive ratio of $\alpha$ for the the Nash social welfare directly implies an $\alpha$ approximation of Pareto efficiency.
%


\citet{BMS19} also consider the sequential allocation of divisible resources, and they impose a normalization on the agent valuations similar to the one in \cite{GPT21}. The main differences are that they enforce envy-freeness as a constraint, and instead aim to maximize social welfare, and, more importantly, that they assume the valuation vectors of the agents are drawn from a distribution. Using this assumption, they weaken the fairness envy-freeness constraint to be satisfied only in expectation.
\citet{HPPZ2019} also study a similar setting, but allow the reallocation of some of the previously allocated items, and show that envy-freeness up to one item can be achieved using $O(T)$ re-allocations. Other works consider the question of online fair division with random agent valuations, e.g. \citet{sinclair2020sequential} look at envy-freeness and efficiency with respect to the hindsight optimal solution, while
\citet{KPW16} study the maxmin share in such a setting. Another line of work on dynamic fair division, which is not very closely related to our work, considers settings where it is the agents, instead of the items, that arrive (and possibly depart) online, e.g., \cite{FPV17,FPS15,LLY18}.

Finally, a line of recent work has focused on approximating the Nash social welfare objective in offline settings. It has played a central role in the literature on the fair allocation of indivisible items, which has focused on the computational complexity of optimizing this objective. This problem is APX-hard \citep{L17}, but recent work has led to a sequence of non-trivial approximation algorithms~(e.g., \cite{CG15,CDGJMVY17,AGSS17,AMGV18,GHM18,BKV18,GKK20}). Also, a very influential paper by \citet{CKMPS19} showed that maximizing the NSW objective when allocating indivisible items leads to approximate envy-freeness, which provided additional motivation for studying this objective. 
Beyond the literature on indivisible items, approximations of this objective have also been studied in strategic settings, where the values of the agents are private information. In this case, research has focused on the design of truthful mechanisms that approximate the NSW \citep{CGG13} or the analysis of non-truthful mechanisms with respect to their price of anarchy with respect to the NSW~\citep{BGM17}. In contrast to this work, in this paper we assume that the agent valuations are public and focus on the complications introduced by the online nature of the problem.

\section{Setting}
\label{sec:setting}


A set of $N$ agents compete for items over $T$ rounds, with a single divisible item arriving in each round. Each agent $i$ has per-unit value $v_{i,t}\geq 0$ for the item in round $t$. Let $V_i = \sum_{t=1}^T v_{i,t}$ denote the sum of agent $i$'s values over all the rounds.
We denote by $x_{i,t}\geq 0$ as the fraction of the item allocated to the agent $i$ on round $t$,  with $\sum_i x_{i,t}=1$ for each round $t$. Under the overall allocation $\bx=\{x_{i,t}\}_{i\in [N], \,t\in [T]}$, the total utility of agent $i$ is given by $u_i(\bx) =  \sum_{t=1}^T v_{i,t} x_{i,t}$. While the valuations $v_{i,t}$ can be arbitrary, we assume that for each agent, we are given a \emph{prediction} $\wtd{V}_i$ of $V_i$. For each $i$, we have $\wtd{V_i}\in \left[\frac{1}{d_i} V_i, ~ c_i V_i\right]$, where $c_i\geq 1$ and $d_i\geq 1$ denote the multiplicative factors by which the prediction $\wtd{V}_i$ may overestimate or underestimate, respectively, the value of $V_i$.



In each round $t$, an online algorithm is a mapping from history of previous rounds, the agents' values on the current round $\{v_{i,t}\}_{i\in [N]}$, and predicted total values $\wtd{V}_i$ to an allocation $x_t=\{x_{i,t}\}_{i\in [N]}$ to be made on this round. Since we assume the items are divisible and the utilities of agents are linear, we can without loss of generality focus on deterministic algorithms.

We evaluate the quality of the final allocation $\bx$ of an online algorithm using two widely studied objectives: the maxmin welfare and the Nash social welfare. 
Under allocation $\bx$, the maxmin welfare (MW) corresponds to the minimum utility across all agents under $\bx$, while the Nash social welfare (NSW) is defined as the geometric mean of the agents' utilities. Formally, we have: 
\begin{align*}
    \text{MW}(\bx)=\min_i\{u_i(\bx)\} \qquad,\qquad \text{NSW}(\bx) = \left ( \prod_i u_i(\bx)\right )^{1/N}
\end{align*}



The performance of our algorithms is measured in terms of their \emph{competitive ratio} with respect to each of these objectives. 
Let $\hat \bx(\bv)$ denote the allocation that the algorithm outputs on an instance $\bv := \{v_{i,t}\}_{i \in [N], \, t \in [T]}$, and let $\bx^{\text{MW}}(\bv)$ and $\bx^{\text{NSW}}(\bv)$ denote the optimal allocation for instance $\bv$ with respect to the maxmin and Nash social welfare objectives respectively. Then the competitive ratio of this algorithm with respect to the two objectives is defined as
\begin{align*}
\gamma^{\text{MW}} = \max_{\bv}\frac{\text{MW}(x^{\text{MW}}(\bv))}{\text{MW}(\hat x(\bv))} \qquad,\qquad 
\gamma^{\text{NSW}}= \max_{\bv}\frac{\text{NSW}(x^{\text{NSW}}(\bv))}{\text{NSW}(\hat x(\bv))}.
\end{align*}

\section{Warm-Up: Naive Attempts at Online NSW Maximization}
\label{sec:challenges}

Before presenting our \algoname\ algorithm and the guarantees for its competitive ratio, we briefly provide some insights regarding the difficulties that arise when trying to maximize NSW in an online setting. We do so by exhibiting some examples of natural, simple algorithms that fail to achieve better than a polynomial competitive ratio for NSW.
For simplicity, the examples in this section assume access to perfect predictions, and that $V_i = 1$ for all $i$.


Arguably the simplest algorithm for allocating items in a balanced fashion is the \emph{uniform allocation}, which sets $x_{i,t}=1/N$ for all $i\in[N]$ and $t\in[T]$; this results in final utilities $u_i(\bx)=V_i/N$. Note that this policy does need access to predicted monopolist utilities $\wtd{V}_i$. However, the resulting competitive ratio is $\Omega(N)$: for example, if $T = N$ and $v_{i,t}=\mathds{1}_{\{i=t\}}\,\forall\,i,t$, then the optimal algorithm gives each agent a final utility of 1. 

An alternative to uniform allocation is \emph{proportional allocation}, which in a setting where all agents have equal predicted monopolist utilities (i.e., $\wtd{V}_i=1$) sets $x_{i,t}={v_{i,t}}/{\sum_j v_{j,t}}$ for all $i\in[N]$ and $t\in[T]$, resulting in utility profile $u_i(\bx) = \sum_{t\in[T]}({v_{i,t}^2}/{\sum_{j}v_{j,t}})$. Assuming perfect predictions and symmetric agents (i.e.,  $\wtd{V}_i=V_i=1\,\forall\,i$), this enjoys the following guarantee:
\begin{proposition}[Proved in~\cref{sec:proportional_sharing}]
\label{prop:propalloc}
If $\wtd{V}_i=V_i=1$ for all $i$, the utility profile under proportional allocation (weakly) Pareto dominates the utilities under uniform allocation. 
\end{proposition}
More generally, in \cref{sec:proportional_sharing}, we show that given (perfect) predictions $\wtd{V}_i=V_i$ for all $i$, the normalized proportional allocation rule $x_{i,t}=\frac{{v_{i,t}/\wtd{V}_i}}{{\sum_j v_{j,t}/\wtd{V}_j}}$ Pareto dominates the utilities under uniform allocation. 
Nevertheless, even when all $V_i=1$ and with perfect predictions $\wtd{V}_i=V_i=1$ for all $i$, proportional allocation fails to get a competitive ratio better than $O(\sqrt{N})$:
\begin{example}[Proportional allocation fails to achieve high NSW]
Consider an instance with $T=N$ rounds and $V_i=1$ for all $i$, where in each round $t$, the `corresponding' agent $i=t$ has a value of $v_{i,t}=\totalvalue/\sqrt{N}$, and every other agent has a value of $v_{i,t}=(\totalvalue-\totalvalue/\sqrt{N})/(N-1)$.
Since $\sum_iv_{i,t}=1$ for all $t$,  under proportional allocation agent $i$ gets $u_i(\bx)=\sum_tv_{i,t}^2 = \frac{1}{N}\left(1+\frac{\sqrt{N}-1}{\sqrt{N}+1}\right)\leq \frac{2}{N}$; this is also the Nash social welfare since all the agents receive the same utility. On the other hand, setting $x_{i,t}=\mathds{1}_{\{i=t\}}$ results in $u_i(\bx)=1/\sqrt{N}$ (and hence also the NSW), which gives a competitive ratio of $\sqrt{N}/2$.
\end{example}

Finally, another natural approach is an online greedy algorithm, which for each round $t$, chooses allocations $\{x_{i,t}\}_{i\in[N]}$ to maximize the Nash social welfare at the end of that round. Surprisingly, this turns out to perform as poorly as uniform allocation.
\begin{proposition}[Proved in~\cref{sec:puregreedy}]\label{prop:myopic_greedy}
Myopic greedy has $\Omega(N)$ competitive ratio.
\end{proposition}
The intuition behind this is that by giving higher priority to agents whose current utility is low, myopic greedy can end up allocating items to agents with very small values in the current round, while ignoring agents with substantially larger values.


\section{\algoname\ Algorithm for Online NSW Maximization}
\label{sec:positiveresults}

We now present our main positive result: the \algoname\ algorithm. The competitive ratio of the algorithm is given in \Cref{thm:main_result}.
We begin with an informal description of the algorithm. \algoname\ divides every item in half, and uses a different strategy for allocating each half:
\begin{itemize}[nosep,leftmargin=*]
\item The \emph{set-aside half} is distributed uniformly among the agents, so in each round every agent receives a $\frac{1}{2N}$ fraction of the whole item -- we will also refer to this as the \emph{promised share} of each agent. 
This ensures that at completion, the utility of agent $i$ is at least $\frac{V_i}{2N}$. 
\item The \emph{greedy half} is allocated so as to maximize the NSW at the end of each round, while incorporating a prediction of the utility that each agent will receive from their promised share. Since the algorithm has access to a prediction $\wtd{V}_i$ of the monopolist utility $V_i$ for each agent $i$, a natural proxy is to use $\frac{\wtd{V}_i}{2N}$ as the prediction of the utility that agent $i$ will receive from the promised share.
\end{itemize}
Note that the allocation rule for the greedy half here is different from the myopic greedy algorithm discussed in~\cref{sec:challenges}; in myopic greedy, the allocation in round $t$ is made so as to maximize the NSW at the end of the round, while in \algoname\ the greedy rule incorporates the prediction that each agent will also receive a $\frac{\wtd{V}_i}{2N}$ utility from the set-aside half. This  circumvents the issues suffered by myopic greedy discussed in~\cref{sec:challenges}.
Also, note that \algoname\ does not need to know the value of $T$.

To define the algorithm formally, we first introduce some additional notation. For each $i \in [N]$ and $t \in [T]$, let $y_{i,t}$ and $z_{i,t}$
be the semi-allocations to agent $i$ in round $t$ from the first and second halves of the item, respectively, so $x_{i,t} = y_{i,t}+z_{i,t}$, and $\sum_i y_{i,t} = \sum_i z_{i,t} = 1/2$ for all $t \in [T]$. 
\begin{definition}[Predicted Utility under Promised Share]
\label{def:promised_utility}
Given round $t\in[T]$, current allocation $z_t=\{z_{i,t}\}_{i\in[N]}$, and historical semi-allocations $z_{t'}= \{z_{i,t'}\}_{i\in[N]}\,\forall\,t'<t$, the \emph{predicted utility under promised share} of agent $i$ is defined as:
\begin{equation*}
\tilde u_{i,t}( z_1,  z_2, \ldots,  z_{t-1}, z_t) ~=~
\frac{\wtd{V}_i}{2N}+ z_{i,t}v_{i,t} + \sum_{t'=1}^{t-1} z_{it'}v_{it'} .
\end{equation*}
\end{definition}
The name of this quantity comes from the following observation: Suppose half the item is distributed uniformly as a promised share (that is, $ y_{i,t} = \frac{1}{2N}$). Then, under any allocation rule $\{z_{i,t}\}$ for the other half, $\tilde{u}_{i,t}$ is a prediction of the final utility of agent $i$, incorporating the knowledge that they will receive a $\frac{1}{2N}$ fraction of the item in each round. 
For brevity, throughout the paper we abuse notation and use ``predicted utility" in place of ``predicted utility under promised share". 
Though predicted utilities depend on all the previous semi-allocations, for notational simplicity we henceforth use the shorthand $\tilde u_{i,t}(z_t)$, dropping the dependence on previous allocations.

\Cref{alg:\algoname} presents the full details of the \algoname\ algorithm. We henceforth use $\hat x, \hat z, \hat y$ exclusively to denote the allocations and semi-allocations produced by this particular algorithm. Half the item is allocated uniformly ($\hat y_{i,t}=1/2N$), and the other semi-allocation $\hat z_t$ is chosen to maximize the Nash social welfare \emph{with respect to the predicted utilities} $\tilde u_{i,t}(z_t)$. Note that the semi-allocation $\hat z_t$ is written in terms of the equivalent \emph{Eisenberg-Gale convex program}: $\max_{z_t} \sum_i \log(\tilde u_{i,t}(z_t))$. While this can be solved in polynomial time via any standard convex optimization solver, in~\cref{sec:tractability} we present a simpler and more efficient way of performing this step, which leads to an overall running time of $O(TN)$. 

\vspace{3pt}
\begin{algorithm}[H]
\caption{\algoname\ algorithm}
{\bf Input:} A prediction $\wtd{V}_i$ of $V_i$ for all $i$.
\begin{algorithmic}[1]
\label{alg:\algoname}
\FORALL{$t=1$ to $T$}
\STATE Set-aside semi-allocation: for each agent $i$ set $\hat y_{i,t}=1/(2N)$.
\STATE Greedy semi-allocation: compute $\hat z_t = \arg\max_{z_t}\left\{\sum_i \log(\tilde u_{i,t}(z_t))\right\}$ ~subject to~ $\sum_i z_{i,t} \leq \frac 1 2,\;z_{i,t}\geq 0$. \\
\COMMENT{Note that $\tilde{u}_{i,t}$ depends on $\wtd{V}_i$.}
\STATE Allocate $\hat x_{i,t} = \hat y_{i,t}+\hat z_{i,t}$.

\ENDFOR
\end{algorithmic}
\end{algorithm}

\vspace{3pt}

We are now ready to state the main result. Recall from \Cref{sec:setting} that $c_i, d_i \geq 1$ are defined to capture the multiplicative error in the predicted total values:   $\wtd{V}_i \in \left[\frac{1}{d_i} V_i, ~ c_i V_i\right]$,  for all $i$.
\main*

 
 
The rest of the section is devoted to proving the result above. We start by describing a duality-based approach for constructing an upper bound on the competitive ratio of a given allocation. We then leverage this technique to prove the theorem. 
 
\subsection{A Duality-Based Upper Bound for the Competitive Ratio}
\label{sec:dual}

In proving the main result of this section, we first construct a dual certificate that tracks the quality of the NSW of the allocation over the execution of the algorithm. As a consequence of this certificate, we get that the \algoname\ algorithm can be alternately interpreted as greedily minimizing the sum of \emph{predicted prices}, which serve as lower bounds on the dual solution.
Informally, the following lemma states that for any allocation, one can construct a vector of ``prices'', one for each round, such that the average price bounds the competitive ratio of the allocation.

\begin{lemma}
\label{lem:prices}
We are given values $\bv$, NSW maximizing allocation $\bx^*$, and allocation $\tilde \bx$. Then for any `price' vector $\{p_t\}_{t\in[T]}\in\mathbb{R}_+^T$ such that $p_t\geq \frac{v_{i,t}}{c_iu_i(\tilde \bx)}$ for every agent $i\in[N]$ and $t\in[T]$, we have
\begin{equation*}
\frac{\mathrm{NSW}(\bx^*)}{\mathrm{NSW}(\tilde \bx)} ~\leq~ \left(\prod_{i=1}^n c_i\right)^{\frac{1}{N}}
\frac{\sum_{t=1}^T p_t}{N}.
\end{equation*}
\end{lemma} 
\begin{proof}
Consider the value profile $\bv$ and the allocation $\tilde \bx$, and let $\pv$ be a vector of prices that satisfies the constraints in the statement of the lemma. We can now write $\pv$ as a solution to a linear program $(P)$, with corresponding dual program $(D)$:
\begin{equation*}
\begin{aligned}[c]
\label{eq:price_lp}
(P) \qquad&\min_{p\in \mathbb{R}_+^T}\ \sum_{t=1}^T p_t \\
&\text{s.t.}\ p_t\geq \frac{v_{i,t}}{c_iu_i(\tilde \bx)}\quad\; \forall\,i\in[N],t\in[T]\\
\end{aligned}
\qquad \qquad
\begin{aligned}[c]
(D)\qquad 
&\max_{\bx\in\mathbb{R}_+^{NT}}\ \sum_{i=1}^N \frac{u_i( \bx)}{c_iu_i(\tilde \bx)} \\
&\text{s.t.}\ \sum_{i=1}^N x_{i,t} \leq 1 \quad \forall\,t\in[T]\\
\end{aligned}
\end{equation*}
Here we define $u_i( \bx) = \sum_{t=1}^T  x_{i,t} v_{i,t}$. We stress that only $ \bx$ (and not $\tilde \bx$) is a variable in the dual program.
By LP duality, any feasible solution $\mathbf{p}$ for the primal program $(P)$ gives an upper bound for the value of the dual program $(D)$. Substituting $\bx=\bx^*$, the optimal NSW solution, we get
$$\frac{1}{N} \sum_{i=1}^N \frac{u_i( \bx^*)}{c_iu_i(\tilde \bx)}  \leq \frac {\sum_t p_t}{N}.$$
Finally, via the AM-GM inequality, we have 
\begin{equation*}
\frac{1}{N} \sum_{i=1}^N \frac{u_i( \bx^*)}{c_iu_i(\tilde \bx)} 
\geq \left [\prod_{i=1}^N \frac{u_i( \bx^*)}{c_iu_i(\tilde \bx)} \right ]^{1/N}  
= \left( \prod_{i=1}^n \frac{1}{c_i}\right)^{\frac{1}{N}} \frac{\text{NSW}(\bx^*)}{\text{NSW}(\tilde \bx)}.\qedhere
\end{equation*}
\end{proof}
Given this lemma, it is useful to re-interpret the \algoname\ algorithm in terms of dual prices. 
Suppose, in light of the lemma above, one wants to achieve a good competitive ratio by choosing an allocation $\tilde{\bx}$ that minimizes prices $p_t$. The problem is, the right side of the inequality $p_t\geq \frac{v_{i,t}}{c_iu_i(\tilde{\bx})}$, namely total utility $u_i(\tilde{\bx})$, depends on future values and allocations, which are unavailable to an online algorithm.

A natural fix to this problem is to construct a lower bound for $c_iu_i(\tilde{\bx})$, and a corresponding upper bound for the feasible price $p_t$. This is exactly what the previously defined predicted utility $\tilde u_{i,t}(z_{i,t})$ does. This is because  in any allocation ${\tilde{\xv}}$, as long as we commit to allocating half of the item uniformly (i.e. $y_{i,t}=\frac{1}{2N}$), we have the following  bound:
$$ 
\tilde u_{i,t}(z_{i,t})
=\frac{\wtd{V}_i}{2N}+ z_{i,t}v_{i,t} + \sum_{t'=1}^{t-1} z_{it'}v_{it'}
\leq \frac{c_iV_i}{2N} + \sum_{t=1}^Tz_{i,t}v_{i,t}
\leq c_iu_i(\tilde{\xv}).
$$
Let us call the corresponding upper bound for the price, $\tilde p_t(z_t) = \max_i \frac{v_{i,t}}{\tilde u_{i,t}( z_{i,t})}$, the \emph{predicted price}. Since it does not depend on future rounds, one can alternatively choose semi-allocations $z_t$ in order to minimize the predicted price $\tilde p_t( z_{t})$ on this round. In fact, doing so leads to exactly the same allocation as \algoname, as we show with the following lemma.
\begin{lemma}
\label{lem:endowed_greedy_dual}
For the allocation $\hat z_t$ chosen by \algoname\ in round $t$,
for any agent $i$ with allocation $\hat z_{i,t}> 0$, we have $\frac{v_{i,t}}{\tilde u_{i,t}(\hat z_{i,t})} = \max_{i'}\frac{v_{i',t}}{\tilde u_{i',t}(\hat z_{i,t})} = \tilde p_t(\hat z_{i,t})$.
\end{lemma}
\begin{proof}[Proof of Lemma \ref{lem:endowed_greedy_dual}]

Fix the round $t$ of the algorithm's execution, and the values $v_{i,t}$. Note that by definition, the predicted utility of agent $i$ at the beginning of round $t$ (before any items are allocated), is $\tilde u_{i,t}(0)$. We prove the lemma by reducing step $t$ of the online allocation to a static Fisher market and showing that the statement of the lemma is equivalent to the KKT conditions for this market. 
    
Consider the following Fisher market: there are $N$ agents and $M=N+1$ items. Item $i$ is only valued by agent $i$ at $v^{\mathrm{Fisher}}_{i,i} = \tilde u_{i,t}(0)$ (i.e. their predicted utility at the beginning of round $t$); agents $i'\neq i$ have zero value for item $i$, that is, $v^{\mathrm{Fisher}}_{i',i}=0$. Agents value the last item at $v^{\mathrm{Fisher}}_{i,M}$ = $v_{i,t}/2$. Consider a static allocation $x$ of items to agents, where $x_{i,j} \in [0,1]$ is the fraction of item $j$ that is allocated to agent $i$. 
Then, the utility of agent $i$ under $x$ is $ x_{i,i}\tilde u_{i,t}(0)+\frac{x_{i,M}v_{i,t}}{2}$. 
    
    Consider the problem of maximizing Nash social welfare in this Fisher market. Clearly, it is optimal to allocate items $1$ to $N$ to the only agent who wants it ($x_{i,i}=1$), and thus it only remains to determine the values of $(x_{i,M}: i \in [N])$. Setting $x_{i,i} = 1$ for all $i \in [N]$, the values of $x_{i,M}$ that maximize Nash social welfare are the solution to the following convex optimization problem:
    \begin{align*}
        \max \; \sum_i \log\left(\tilde u_{i,t}(0)+\frac{x_{i,M}v_{i,t}}{2}\right) 
        \quad  
        \text{subject to \;\; $\sum_i x_{i,M} \leq 1$, and
        $x_{i,M}\geq 0$ for all $i \in [N]$.}
    \end{align*}

 Changing variables $z_{i} = \frac{x_{i,M}}{2}$ leads to the optimization problem identical to the one in \algoname. This implies the equivalence between the online solution $\hat z_t$ and the optimal solution $x^*$ for the constructed Fisher market:  $\hat z_{i,t} = x_{i,M}^*/2$.
On the other hand, KKT conditions for this static Fisher market give (\cite[Chapter 5]{AGT}):

$$x_{i,M}^*\neq 0 \Rightarrow \frac{v_{i,t}}{\tilde u_{i,t}(0)+\frac{v_{i,t}x_{i,M}^*}{2}} = \max_{i'} \frac{v_{i',t}}{\tilde u_{i',t}(0)+\frac{v_{i',t}x_{i',M}^*}{2}}$$

Now, substituting $\hat z_{i,t} = \frac{x_{i,M}^*}{2}$ yields the statement of the lemma.
\end{proof}

\subsection{Bounding the Competitive Ratio of \algoname}

Consider the allocations $\hat \bz$ made under \algoname, and the corresponding predicted prices $\tilde p_t (\hat z_t) = \max_i \frac{v_{i,t}}{\tilde u_{i,t}(\hat z_t)}$.
These prices satisfy the condition of~\cref{lem:prices}, since we have from the definition of the predicted prices that, for all $i$, $\tilde p_t (\hat z_t) = \max_{i'} \frac{v_{i',t}}{\tilde u_{i',t}(\hat z_t)}\geq \frac{v_{i,t}}{c_iu_i(\hat \bx)}$. 
To complete the proof of~\cref{thm:main_result}, we need to argue that $\sum_{t=1}^T \tilde p_t(\hat z_t) $ is small. 

\begin{lemma}
\label{lem:promised_price_one_round}
Under \algoname, in every round $t$ we have
$$ \tilde{p}_{t}(\hat{z}_t) \leq 2\sum_{i=1}^N \left[ \ln\left(\tilde{u}_{i,t}(\hat{z}_t)\right) - \ln\left(\tilde{u}_{i,t}(0)\right)\right].$$
\end{lemma}
\begin{proof}
We have 
\begin{align*}
    \frac12\tilde{p}_{t}(\hat{z}_t)
    &= \sum_{i=1}^N \hat{z}_{i,t}\tilde{p}_{t}(\hat{z}_t) &\text{(since $\sum_i \hat{z}_{i,t} = \frac12$ for all $t$)} \\
    &= \sum_{i=1}^N  \hat{z}_{i,t}\cdot\frac{v_{i,t}}{\tilde{u}_{i,t}(\hat{z}_t)} &\text{(by Lemma \ref{lem:endowed_greedy_dual})} \\
    &=\sum_{i=1}^N \frac{\tilde{u}_{i,t}(\hat{z}_t) - \tilde{u}_{i,t}(0)}{\tilde{u}_{i,t}(\hat{z}_t)} &\text{(by definition of $\tilde{u}_{i,t}$)} \\
    &=\sum_{i=1}^N \left(1 - \frac{\tilde{u}_{i,t}(0)}{\tilde{u}_{i,t}(\hat{z}_t)}\right) \\
    &\leq \sum_{i=1}^N -\ln \left(\frac{\tilde{u}_{i,t}(0)}{\tilde{u}_{i,t}(\hat{z}_t)}\right) &\text{(using the inequality $1-x \leq - \ln x$)}\\
    &= \sum_{i=1}^N \left[\ln\left(\tilde{u}_{i,t}(\hat{z}_t)\right) - \ln\left(\tilde{u}_{i,t}(0)\right)\right].\tag*{\qedhere}
\end{align*}
\end{proof}

~\Cref{thm:main_result} is now a simple corollary of this bound.

\begin{lemma}
\label{lem:final_price_bound}
For the allocation $\hat z$ made by \algoname, the following holds:
\begin{align*}
\frac{\sum_{t=1}^T \tilde p_t(\hat z_t)}{N} 
&\leq \min\left\{\log(2N) + \frac{1}{N}\sum_{i=1}^N\log(d_i), \, \log(2T) + \log(\max_i\{d_i\}) \right\}.
\end{align*}
\end{lemma}
\begin{proof}[Proof of Lemma \ref{lem:final_price_bound}]

Using \Cref{lem:promised_price_one_round}, we get
\begin{align*}
    \sum_{t=1}^T \tilde{p}_t 
    &\leq \sum_{i=1}^N \sum_{t=1}^T \left(\log(\tilde{u}_{i,t}(\hat{z}_t)) - \log(\tilde{u}_{i,t}(0)) \right) \\
    &= \sum_{i=1}^N \left( \log(\tilde{u}_{i,T}(\hat{z}_T)) - \log(\tilde{u}_{i,1}(0))\right)\\
    &= \sum_{i=1}^N \left( \log(\tilde{u}_{i,T}(\hat{z}_T))) - \log\frac{\wtd{V}_i}{2N}\right)\\
    &= N\log(2N) + \sum_{i=1}^N \log\left(\frac{\tilde{u}_{i,T}(\hat{z}_T)}{\wtd{V}_i}\right).
\end{align*}
Observe that 
$$\tilde{u}_{i,T}(\hat{z}_T) 
= \frac{\wtd{V}_i}{2N} +\sum_{t=1}^T v_{it}\hat{z}_{it} 
\leq \frac{\wtd{V}_i}{2N} + \left(\sum_{t=1}^T v_{it}\right) \max_t \hat{z}_{it}
= \frac{\wtd{V}_i}{2N} + \frac12 V_i
\leq  \frac{\wtd{V}_i}{2N} + \frac12 d_i\wtd{V}_i.$$
Hence, 
$$
\frac{\tilde{u}_{i,T}(\hat{z}_T)}{\wtd{V}_i}
\leq \frac{1}{2N} + \frac{d_i}{2} \leq d_i.
$$
This gives $\sum_{t=1}^T \tilde{p}_t \leq N\log(2N) + \sum_{i=1}^N \log(d_i)$, which proves the first part of the claimed bound.

On the other hand, for $d = \max_i\{d_i\}$, we have the bound
\begin{align*}
    \sum_{i=1}^N \log \left( \frac{\tilde{u}_{i,T}(\hat{z}_T)}{\wtd{V}_i} \right)
    &= N\log\left( \left(\prod_{i=1}^N \frac{\tilde{u}_{i,T}(\hat{z}_T)}{\wtd{V}_i}\right)^{\frac{1}{N}}\right) \\
    &\leq N\log\left( \frac{1}{N} \sum_{i=1}^N \frac{\tilde{u}_{i,T}(\hat{z}_T)}{\wtd{V}_i}\right) \\
    &= N\log\left( \frac{1}{N} \sum_{i=1}^N\frac{1}{\wtd{V}_i} \left( \frac{\wtd{V}_i}{2N} + \sum_{t=1}^T v_{it}{\hat{z}_{it}} \right)\right) \\
    &= N\log\left( \frac{1}{2N} + \frac{1}{N} \sum_{t=1}^T \sum_{i=1}^N \frac{ v_{it}{\hat{z}_{it}}}{\wtd{V}_i}  \right) \\
    &\leq N\log\left( \frac{1}{2N}+ \frac{1}{N} \sum_{t=1}^T \frac12\max_i\left\{ \frac{v_{it}}{\wtd{V}_i}\right\}\right)\\
    &\leq N\log\left( \frac{1}{2N}+ \frac{Td}{2N}\right)\\
    &\leq N\log\left( \frac{Td}{N} \right)
\end{align*}
This gives $\sum_{t=1}^T \tilde{p}_t \leq N\log(2N) + N\log\left( \frac{Td}{N} \right) = N \log(2Td)$.




\end{proof}

\begin{proof}[Proof of \Cref{thm:main_result}]
The theorem follows directly from~\cref{lem:prices} and \cref{lem:final_price_bound}, as the predicted prices $\tilde p_t(\hat z_t)$ satisfy the inequalities in the statement of \cref{lem:prices}.
\end{proof}


\section{Lower Bounds}
\label{sec:lowerbnds}

In this section, we complement our positive result by showing that the guarantee achieved by the \algoname\ algorithm is tight up to sub-logarithmic factors.
To build some intuition into the hardness of maximizing NSW, we first show that for a closely related problem of MW maximization, \emph{no algorithm can have competitive ratio of $O(\sqrt{N})$}. 
\begin{theorem}
\label{thm:impossibility_maxmin}
No online algorithm can achieve a competitive ratio better than $O(\sqrt{N})$ with respect to the maxmin welfare, even assuming symmetric agents and perfect predictions (i.e., $\wtd{V}_i = V_i=1\,\forall\,i$), and that all $v_{i,t}\in \left\{0, \frac{1}{2}\right\}$.
\end{theorem}

This shows that MW is a problematic objective in online fair allocation as it does not admit a good competitive ratio even with perfect predictions. Moreover, the proof of this result provides insight into the more complex lower bound construction for online NSW maximization.

\begin{theorem}\label{thm:main_negative_result}
No online algorithm can achieve a competitive ratio of $O(\log^{1-\epsilon}{N})$ or $O(\log^{1-\epsilon}{T})$ for a constant $\epsilon>0$ with respect to the Nash social welfare, even if it knows that $V_i=1$ for all $i$.
\end{theorem}

Both these results are based on hard instances which are similar in spirit: First, at some time $t$ the algorithm is confronted by agents with identical values; next, in later rounds, some of these agents are able to obtain large values without competition (i.e., there are rounds when only one agent has a non-zero value), while other agents clash over a single item, and are thus unable to get much value from future rounds. Note that in both cases, the online algorithm has perfect predictions of the monopolist utilities $V_i$, and moreover, ex post the agents are symmetric (in that they all have $V_i=1$). The hardness arises instead from the fact that an online algorithm is unable to predict which agents will clash in future.

\subsection{Hardness of Approximating Online Maxmin Welfare}
\label{ssec:lowerbndminmax}
We first give the lower bound construction for \cref{thm:impossibility_maxmin}, which, as we mentioned, exposes the main idea in the lower bound construction for NSW in~\cref{sec:hardnessNSW}: when facing identical agents in earlier rounds, an online algorithm can do no better than treat them symmetrically; however, some of these agents may later have high value in rounds with low competition, while others may only care for items that are highly competed for, making them harder to satisfy. An algorithm with access to future valuations can prioritize the latter agents in earlier rounds, which no online algorithm can  hope to achieve. Adapting this for NSW, however, is much more complicated. 
\begin{figure*}[!t]
{ $$
V = \begin{bmatrix}
{\color{red}\mathbf{1/2}} & 1/2 & 1/2 & \cdot & \cdot & \cdot & \cdot & \cdot & \cdot \\
\cdot & \cdot & 1/2 & \cdot & \cdot & \cdot & \cdot & \cdot & \cdot\\
\cdot & 1/2 & \cdot & \cdot & \cdot & \cdot & \cdot & \cdot & \cdot\\
\cdot & \cdot & \cdot & {\color{red}\mathbf{1/2}} & 1/2 & 1/2 & \cdot & \cdot & \cdot\\
\cdot & \cdot & \cdot & \cdot & 1/2 & \cdot & \cdot & \cdot & \cdot\\ 
\cdot & \cdot & \cdot & \cdot & \cdot & 1/2 & \cdot & \cdot & \cdot\\
\cdot & \cdot & \cdot & \cdot & \cdot & \cdot & 1/2 & {\mathbf{\color{red}1/2}} & 1/2\\
\cdot & \cdot & \cdot & \cdot & \cdot & \cdot & \cdot & \cdot  & 1/2 \\
\cdot & \cdot & \cdot & \cdot & \cdot & \cdot & 1/2 & \cdot & \cdot \\
{\mathbf{\color{red}1/2}} & \cdot & \cdot & {\mathbf{\color{red}1/2}} & \cdot & \cdot & \cdot & {\color{red}\mathbf{1/2}} & \cdot
\end{bmatrix}
$$}
\caption{\it\small Lower bound construction for maxmin welfare from~\cref{thm:impossibility_maxmin} with $N=9$ and $T=10$. Each row is a round and each column is an agent. Values of $0$ are indicated with `\,$\cdot$', and the special agent in each group of $3$ is highlighted in bold red.}
\label{fig:eswlower}
\vspace{-0.4cm}
\end{figure*}

\begin{proof}[Proof of~\cref{thm:impossibility_maxmin}]
We construct a family of instances where $N$ is the square of an integer, and $T=N+1$.
The instance with $N=9$ is depicted in~\cref{fig:eswlower}. Agents are grouped into $\sqrt{N}$ groups, each of size $\sqrt{N}$, with one special agent chosen uniformly at random in each group. 
The $T$ rounds are grouped in $\sqrt{N}$ epochs of $\sqrt{N}$ rounds each, plus one final round. In each epoch $i$, only agents in the corresponding group $i$ have non-zero values.
In the first round of epoch $i$, the $\sqrt{N}$ agents in corresponding group $i$ have value $\frac12$; in the subsequent $\sqrt{N}-1$ rounds, each agent in the group except for the special agent \emph{resolves}, i.e., has value $\frac12$ while all other agents have value $0$. 
In the very last round (that is, the $T^{\text{th}}$ round), the $\sqrt{N}$ special agents all arrive simultaneously with value $\frac12$. 

By symmetry, an online algorithm can do no better than allocating uniformly among the agents with nonzero value in each round: This gives a utility of $\frac{1}{2\sqrt{N}} + \frac12$ for every non-special agent, and $\frac{1}{\sqrt{N}}$ for every special agent.
On the other hand, consider an algorithm that allocates the entire item to the special agent in the first round of each epoch, and allocates uniformly among the agents with nonzero value in each other round: This algorithm gets utility at least $\frac12$ for all agents, which immediately gives a $\frac{\sqrt{N}}{2}$ lower bound for the competitive ratio of any online algorithm.
In more detail, one can compute that the optimal offline algorithm allocates $1 + \frac{1}{N} - \frac{1}{\sqrt{N}}$ of the first item in each epoch  to the special agent, resulting in a final utility of $\frac12 + \frac{1}{2N}$ for all agents. Hence, the competitive ratio of any algorithm is bounded below by 
$$\gamma^{\text{MW}} \geq \frac{\frac12 + \frac{1}{2N}}{\frac{1}{\sqrt{N}}} = \frac{\sqrt{N}}{2} + \frac{1}{2\sqrt{N}}.$$
\end{proof}



\subsection{Hardness of Online NSW Maximization}
\label{sec:hardnessNSW}

We first give a description of the lower bound instance used in proving \cref{thm:main_negative_result}.

Let $N$ be large enough such that $M = \log^{(1-\epsilon)} N$ is an integer (indeed, throughout the proof we assume that for large enough $N$ and rational $\epsilon$, all our parameters are integers). Also, let $L$ be an integer constant to be chosen later on. For each pair of  integers $m,\ell$ such that $1\leq m\leq M$ and $1\leq \ell\leq L$, we let $w^{m,\ell} = \frac{1}{\log^{(L-\ell+1)(M-m+1)} N}$. Note that for any constant $L$, we have $w^{m,\ell}=\Omega(1/N)$ and $w^{m,\ell}=o(1/M)$.  Also, for any fixed $\ell$ we have $\frac{w^{m,\ell}}{w^{m+1,\ell}}=\frac{1}{\log N}=o(1/M)$.

\begin{figure*}[t]
\includegraphics[trim=0 300 0 300,clip,width=\textwidth]{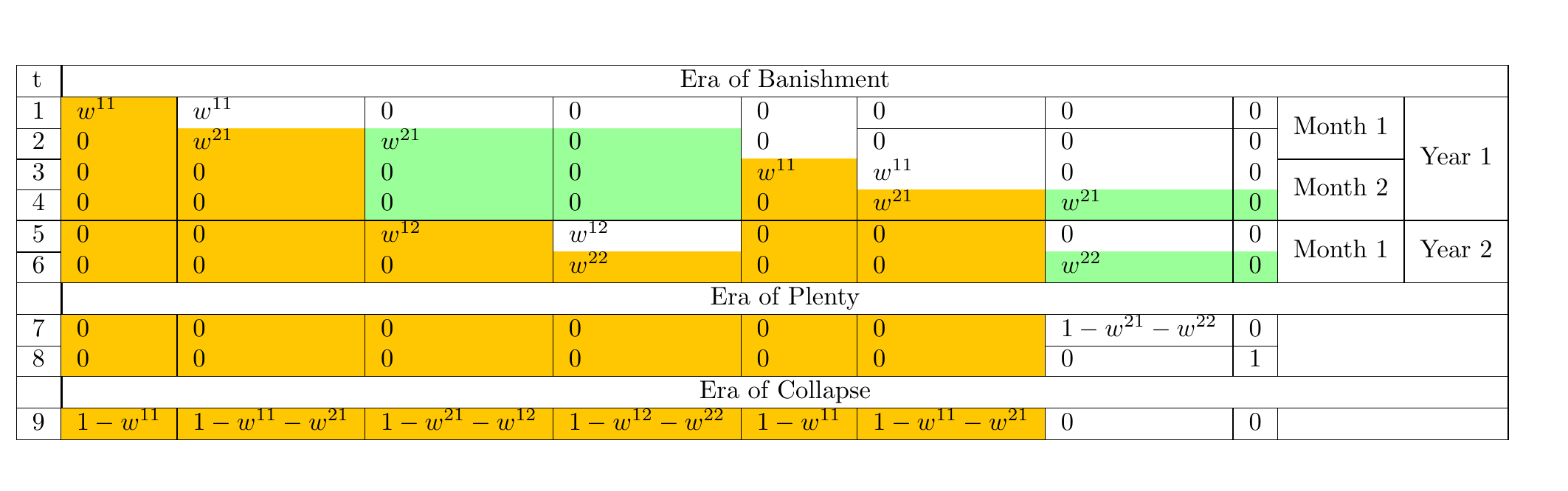}
\caption{\it\footnotesize Competitive-ratio lower bound for $NSW$: illustrating the hard instance from~\cref{thm:main_negative_result}. Values in cells are the values of the agents on different rounds, orange represents the agent being banished, and green being clear.}
\label{fig:lowerbndNSW}
\end{figure*}
 
To make our description more readable, we split up rounds in our instance into a hierarchy of repeating cycles. In more detail, our instance is made up of three \emph{eras}: the Era of Banishment, the Era of Plenty, and the Era of Collapse. The Era of Banishment is further subdivided into \emph{years}, with a year comprising of  \emph{months}. The structure of the instance is illustrated in~\cref{fig:lowerbndNSW}.

\textbf{Era of Banishment}. This era lasts $N(1-1/2^L)$ rounds. It consists of $L$ years (indexed by $\ell \in [L]$), with year $\ell$ lasting $N/2^\ell$ rounds. Each year consists of months (indexed by $m \in [M]$), where each month lasting $M$ rounds. (Note that different years contain different number of months. Namely, year $\ell$ contains $\frac{N}{2^\ell M}$ months.) In each year of the Era of Banishment, agents are split into ``banished'' and ``cleared'', by banishing $M$ agents and clearing $M$ agents per month. Once banished, an agent never has non-zero valuation again until very last round (the Era of Collapse). The cleared agents lose their status at the end of the year (and can become cleared or banished in the next year).

In the first round of month $m$ of year $\ell$, there are $M$ agents who are neither banished nor cleared who have value $w^{1,l}$ (everyone else has value 0). After the algorithm $\mathcal A$ makes an allocation $x_t$, we find an agent for whom $x_{i,t}\leq 1/M$ (such agent always exists by pigeonhole principle) and this agent becomes banished.

On round $z\in[2,M]$ of this month there are $M-1$ agents from previous round plus one new agent (who isn't banished and hasn't been cleared this year) with values $w^{z\ell}$ (remaining agents, including the banished ones, have value 0). After the allocation is made, we repeat the banishing procedure described above. By the end of the month, we have banished $M$ agents, and we declare all agents who had non-zero value this month but weren't banished (there are $M-1$ such agents) to be cleared. To make the accounting easier, we additionally clear 1 more agent who hasn't been cleared or banished before in this year (thus equalizing number of banished and cleared agents in this month). Also note that all agents can have non-zero value during only one month in a year, since during this month they have to become either banished or cleared. 

The number of months in the year is $\frac{N}{2^\ell M}$, as after this many months every agent is either banished or cleared. At the end of the year, the cleared agents lose their status (but banished agents remain such), and the next year begins. 

Thus, every year of this era we banish half the agents that were not banished, and by the end of the era (which lasts $L$ years) we banish $N(1-1/2^L)$ agents.

We also note that very little value is seen by the agents during this Era (compared to their total allowed value of $\sum_t v_{i,t}=1$). Since during the Era of Banishment each agent participates in at most one month a year, the total value seen by an agent during the era of Banishment can be bound by $\sum_{m,\ell} w^{m,\ell}\leq MLw^{M,L} = \frac{L\log^{1-\epsilon}N}{\log N} = o(1)$ (since the number of years $L$ is a constant).

\textbf{Era of Plenty.} This era begins at $t = N(1-1/2^{L})+1$ and lasts $N/2^L$ rounds. At every round there is a single non-banished agent with the remainder of their value $v_{i,t} = 1-\sum_{t'=1}^{N(1-1/2^L)}v_{i,t'}=1-o(1)$ (see the paragraph above for the justification of the last equality). All other agents have zero value on this round).

\textbf{Era of Collapse.} This era consists of just one round, $t=T=N+1$. At this round, all banished agents have the remainder of their value $v_{i,t}= 1-\sum_{t'=1}^{N}v_{i,t'} = 1-o(1)$.

Given the above instance, we are now ready to prove our hardness result.
\begin{proof}[Proof of \cref{thm:main_negative_result}]
Assume that there exists an online algorithm $\mathcal A$ that achieves an competitive ratio of $\log^{{1-\epsilon}}T$ or $\log^{{1-\epsilon}}N$ for some constant $\epsilon \geq 0$. We prove that this leads to a contradiction by showing that $\mathcal A$ would fail to satisfy these guarantees for the instance described above. 

We first provide a lower bound for the optimal Nash social welfare in this instance. During the Era of Banishment we can give the item to the agent who is about to become banished, during the Era of Plenty we give the item to the only agent with non-zero value, during the Era of Collapse we split the item equally between the agents.

In this case, utility of banished agents $u_i^*$ can be lower bounded with the value $w^{ij}$ on the round when they got banished, and utility of non-banished agents is $1-o(1)$, because $\sum_{m\leq M,\ell\leq L} w^{m,\ell} = o(1)$ according to our definition of $w^{m,\ell}$ above.

We now compute utilities under the allocation made by the online algorithm. Every banished agent got $1/M$ fraction of the item on the round that they got banished, and their values on previous rounds is negligible $v_{i,t'}= o(v_{i,t}/M)$, since $\frac{w^{m-1,\ell}}{w^{m,\ell}} = o(1/M)$. 

Finally, we argue that  the last round does not affect the utilities of banished agents much. Intuitively, this is because there is only one item and $O(N)$ banished agents who have approximately equal value for it. We provide a more careful argument below.

We upper bound optimal utilities on the last round by assuming banished agents got $1/M$ of the item on which they got banished (by construction, maximum possible amount), and assuming that the algorithm makes optimal NSW allocation on the last round.
\\
The optimal NSW solution, when choosing between agents with equal values, allocates the marginal fraction of the item to the agent with lowest utility. Note that, on the last round, there are $N/2M$ agents with value $w^{11}$ (the smallest value of $w^{lm}$). Thus, if the item is distributed between them, each agent gets the utility of $\frac{2M}{N}=\frac{2\log^{1-\epsilon}}{N} = o(w^{11}/M)$, thus not improving their utility by even a constant factor ( and so it will distribute the entire item between the agents with $w^{11}$)
. 

Thus, overall, utility of banished agents can be upper-bounded by $u_i \leq \frac{u_i^*}{M}+o(\frac{u_i^*}{M})$. We bound utility of non-banished agents with $u_i\leq 1$.

{Now we can bound the competitive ratio achieved by the online algorithm as follows:}

$$
\frac{\mathrm{NSW}(\xv)}{\mathrm{NSW}(\xv^*)} \leq \left [ \frac{u_i}{u_i^*}\right]^{1/N} \leq \left [ \left(\frac{1}{M}\right)^{N(1-1/2^L)} \left( \frac{1}{1-o(1)}\right)^{N/2^L}\right]^{1/N} \approx \left(\frac{1}{M}\right)^{1-1/2^L} = \log^{1-\epsilon'} N.
$$
Finally, we pick the constant $L$ (and appropriately large $N$) to make sure $\epsilon'\leq \delta$. 
\end{proof}


\section{Discussion}
\label{sec:conclusion}
We conclude with a brief discussion regarding the our model and possible extensions of our results. 

\noindent\textbf{Allocating multiple items per round.}
It is worth noting that all our results readily generalize to a setting where multiple items can arrive on a single round. The reduction is to split a single round with multiple items into multiple ``imaginary'' rounds with a single item in each. Under this reduction, our approximation guarantees stay intact, with number of rounds $T$ replaced by the total number of items.

\noindent\textbf{Allocating indivisible items using randomness.}
Also, note that even if the items are actually \emph{indivisible}, i.e., each item can be allocated to at most one agent, then both our upper bounds and our lower bounds extend to this setting as well by considering randomized algorithms. In this case, we can evaluate the agents' \emph{expected} utility based on the probability that they receive each item, and linearity of expectations reduces this setting to the one studied in this paper.

\section{Additional Proofs and Counterexamples}
\subsection{No Additional Information about Agents' Values}\label{sec:unscaled}
Here we show that if the online algorithm is given no information about the agents' values, then there is no hope for a non-trivial competitive ratio for the NSW in an online setting. It is easy to see that splitting every item equally gives a $1/N$ approximation of NSW. We now prove the following result.

\begin{theorem}
In the absence of any additional information regarding the agents' values, there is no online algorithm that achieves a  $\frac{e}{N}+\delta$ or $\frac{e}{T}+\delta$ competitive ratio with respect to the NSW for a constant $\delta>0$.
\end{theorem}
\begin{proof}
We assume an arbitrary online allocation algorithm $M$ is employed and construct an instance on which it fails to provide a $\frac{e}{N}+\delta$ approximation to NSW.

There are $T=N$ rounds in the instance. At $t=1$ all agents have the value of $v_{i,t} = 1$. We then pick an agent for whom $x_{i,t}\leq 1/N$ (who exists by pigeonhole principle) and that ensure this agent would never have non-zero values for the rest of the instance (we would refer to such agents as frozen).

At any round t we will have non-frozen agent to have values $v_{i,t} = \frac{1}{\epsilon^t}$ ( we will pick a sufficiently small $\epsilon<1$ later on), and freeze the agent whose allocation is $x_{i,t}\leq 1/(N-t)$.

NSW of the optimal allocation $x^*$ can be lower bounded by, on round $t$, allocating all of to the agent who is about to get frozen, this yields $u_i(x^*) \geq 1/\epsilon^{i-1}$ (we rename indices in the order of being frozen).

On the other hand, by construction the algorithm $M$ achieves $u_i \leq \frac{u_i(x^*)}{i}+O(\epsilon
\frac{u_i(x^*)}{i})$. This yields a lower bound for the NSW approximation ratio $\frac{NSW}{NSW^*}\leq (\frac{1}{N!}+O(\epsilon))^{1/N} \approx \frac{e}{N}+O(\epsilon^{1/N})$.

Clearly, choosing a sufficiently small $\epsilon$ would yield the needed lower bound for any $\delta>0$.
\end{proof}

\subsection{Proportional allocation (Proof of~\cref{prop:propalloc})}
\label{sec:proportional_sharing}

In~\cref{sec:challenges} we mentioned that proportional allocation Pareto-dominates uniform allocation in terms of agents' utilities, if $V_i=1$ for all $i$. We now show this by proving a more general result for the normalized proportional allocation rule under perfect predictions of monopolist utilities.
\begin{theorem}
Given perfect predictions $V_i$ of monopolist utilities for each agent $i$, let $\tilde x$ be the allocations made by the (normalized) proportional allocation rule $\tilde x_{i,t} = \frac{v_{i,t}/V_i}{\sum_j v_{j,t}/V_j}$. Then for each agent $i$, we have
$$u_i(\tilde x)\geq \frac{V_i}{N}.$$
\end{theorem}

\begin{proof}
Fix parameters $T,N$, predictions $\{V_j\}$ and for any given agent $i$, consider any sequence of values $v_{i,t}$ such that $\sum_t v_{i,t}=V_i$. 
We prove the statement of the theorem by optimizing the values of all other agents to make the utility of agent $i$ as small as possible, and show that the minimum value is never less than $V_i/N$.
    
Notice that under the Proportional Sharing allocation rule, the allocation $\tilde x_{i,t} = \frac{v_{i,t}/V_i}{v_{i,t}/V_i+\sum_{j\neq i}v_{j,t}/V_j}$ of agent $i$ does not depend on the particular values $v_{j,t}$ that other agents report, but rather, only their normalized sum $\sum_{j\neq i}v_{j,t}/V_j$; moreover, under perfect predictions, we know that for every agent $j$ we have $\sum_tv_{j,t}/V_j=1$. Thus, in minimizing the utility of agent $i$, one can think of agents other than $i$ as a single super-agent who can report values $v_{t}' = \sum_{j\neq i}v_{j,t}/V_j$ that sum up to no more than $N-1$. Then, the problem becomes one of choosing $v_t'$ as to minimize the utility of agent $i$, which is given by the following convex program: 
\begin{align*}
\min_{v'} \sum_t v_{i,t}\frac{v_{i,t}/V_i}{v_{i,t}/V_i+v_t'}\\
s.t. \sum_t v_t' \leq N-1.
\end{align*}
For optimizing over $v_t'$, the first-order conditions are as follows:
\begin{align}
-\frac{v_{i,t}^2/V_i}{(v_{i,t}/V_i+v_t')^2}+\lambda = 0 \quad \forall\, t\\
\sum_t v_t' = N-1
\end{align}
To solve it, given any Lagrange multiplier $\lambda$, we substitute $v_t' = \frac{v_{i,t}}{V_i}\left(\sqrt{\frac{V_i}{\lambda}}-1\right) $, and then from the KKT constraints we get:
$$\sum_t v_t' = \sum_t \frac{v_{i,t}}{V_i}\left(\sqrt{\frac{V_i}{\lambda^{\star}}}-1\right) = \sqrt{\frac{V_i}{\lambda^{\star}}}-1 = N-1$$
which gives $\lambda^{\star} = V_i/N^2$ and $v_t'^{\star} = \frac{v_{i,t}}{V_i}\left(N-1\right)$.
Substituting back in the objective, we see that under the proportional allocation rule $\tilde x$, agent $i$ gets utility at least $\sum_t v_{i,t}\frac{v_{i,t}/V_i}{v_{i,t}/V_i + v_{i,t}(N-1)/V_i} = \frac{V_i}{N}$, which is precisely the utility under the uniform allocation rule.    
\end{proof}

\subsection{Myopic Greedy Algorithm (Proof of~\cref{prop:myopic_greedy})}
\label{sec:puregreedy}

In this section, we show that a simple greedy algorithm that chooses allocation $\{x_{i,t}\}_{i\in[N]}$ so as to maximize the $NSW$ at the end of round $t$, has $\Omega(N)$ competitive ratio. 

As in the previous hardness results, we consider the simplest setting with perfect predictions of monopolist utilities $V_i$, and also that $V_i=1$ for every agent $i$. Let $u_{i,t}(x_t)=\sum_{t'=1} x_{i,t'} v_{i,t'}$ be the utility  of agent $i$ for all rounds up to $t$. The myopic greedy algorithm is formally defined as:
\begin{align*}
    \hat x_t = \arg&\max_{x_t}\sum_i \log(u_{i,t}(x_t))\\
        & \;s.t.\ \sum_i x_{i,t}\leq 1\quad,\quad x_{i,t}\geq 0\quad\forall\,i\in[N]
\end{align*}
Analogous to~\cref{lem:endowed_greedy_dual}, the myopic greedy algorithm can be interpreted as a myopic price-minimizer.
\begin{lemma}
\label{lem:pure_greedy_dual}
Consider allocation $\hat x_t$ made by the myopic greedy algorithm in some round $t$. Then
\begin{align*}
\hat x_{i,t}\neq 0 \Rightarrow \frac{v_{i,t}}{ u_{i,t}(\hat x_t)} =
 \max_{i'}\frac{v_{i't}}
 {\tilde u_{i't}{\hat x_{i,t}}} = \hat p_t
\end{align*}
\end{lemma}
Now consider the following instance with $T = N^2$:
In rounds $t\in[1,N]$ agent $i=1$ has value $v_{1t} = 1/N$; agents other than $1$ have values $v_{i,t} = 1/N^{N^2-t+1}$, which although geometrically growing in $t$, is still vanishingly small compared to the active agent. We refer to these as the `active' rounds for agent $1$.

After $N$ rounds, agent $1$ exhausts her total value. Now for rounds $t\in[N+1,2N]$, agent 2 is active, with value $v_{2t} = \frac{1}{N}-\frac{\sum_{t'<t} v_{2t'}}{N} = \frac{1}{N}-o(\frac{1}{N})$, while agents $i>2$ continue to have value $ v_{i,t} =1/N^{(N^2-t+1)}$. This process repeats, with agent $i$ active in rounds $t\in[(i-1)N+1,iN]$; thus at the end of $N^2$ rounds, each agent is for $N$ rounds.

To bound the competitive ratio of myopic greedy on this instance, first note that the optimal allocation can be lower bounded by always giving the item to the active agent, which yields $u_i(\bx) \geq 1-o(1/N) $, and thus the optimal $NSW\geq 1-o(1/N)$.

On the other hand,~\cref{lem:pure_greedy_dual} implies that in round $t$, the myopic greedy algorithm allocates the item to the agent who maximizes the ratio $\frac{v_{i,t}}{u_{i,t}(x_t)}$. In round $(i-1)N+1$ when agent $i$ is active for the first time, we can upper bound their utility by assuming they obtain the whole item (and thus the value of $k/T-o(1/T)$). On the following rounds the ratio $\frac{v_{i,t}}{u_{i,t}(x_t)} = 1-o(1)$ for the active agents, but $\frac{v_{i,t}}{u_{i,t}(x_t)} \geq N-i$ for the other agents with non-zero value (the bound of $N-i$ comes from allocating all of the items equally between non-active agents with non-zero value). Thus the active agent is not allocated anything on any rounds except for the first one. 

Thus for the utilities of agents in the greedy solution can be bounded with $u_i(\hat x) \leq \frac{1}{N}+o(\frac{1}{N})$. 

This implies an approximation factor $A\geq N$, which is no better than just allocating all items equally between the agents.

\subsection{Computational Tractability of the \algoname\ Algorithm}

\label{sec:tractability}
We now briefly note that the \algoname\ algorithm is quite tractable from a computational standpoint. The most demanding operation of this algorithm involves the computation of $\hat z_{i,t}$ for each round $t$ according to the following program:
\begin{align*}
\hat z_t = \arg \max_{z_t} &\sum \log(\tilde u_{i,t}(z_t))\\
\text{s.t.}\ & \sum_i z_{i,t}\leq \frac{1}{2} \quad , \quad
 z_{i,t}\geq 0.
\end{align*}

Although this is a convex program that can be solved in polynomial time, we can actually also provide a fast alternative process for computing this allocation. \Cref{lem:endowed_greedy_dual} implies that the problem above is equivalent to the problem of minimizing the current promised price:

\begin{align}
\hat z_t = \arg \min_{z_t}\ & p_t \label{eq:comp1}\\
\text{s.t.}\ & \sum_i z_{i,t}\leq \frac{1}{2} \nonumber\\
& z_{i,t}\geq 0 \nonumber\\
&p_t \geq \frac{v_i}{\tilde u_i(z_t)} 
\label{eq:comp4}
\end{align}

We now provide an algorithm to solve this problem. The idea is to allocate the marginal fraction of the item to the agents for whom the constraint \eqref{eq:comp4} is binding, as there is no other way to decrease the objective. We now demonstrate how to allocate the entire item this way $N$ steps. 

We initialize $z_{i,t}=0$ for every agent, and initialize the unallocated  fraction of the item at $B=1/2$. We order the agents in decreasing order on their $\frac{v_{i,t}}{\tilde u_{i,t}(z_{i,t})}$ ratios using their promised utilities. Assuming that the agents are re-indexed according to this order, agent 1 has the highest ratio and agent $N$ has the lowest ratio. For each agent $i<N$, we then compute the value $\delta_i$ which is the solution to the following equation:
$$\frac{v_{i,t}}{\tilde u_{i,t}(\delta_i)} = \frac{v_{i+1,t}}{\tilde u_{i+1,t}(0)}.$$ 
This value, $\delta_i$ captures how much allocation does agent $i$ need to catch up to the next agent in terms of the $\frac{v_{i,t}}{\tilde u_{i,t}(\delta_i)}$ ratio.

Starting with $i=1$, we then execute the following steps. If $i\delta_i\leq B$, we allocate $\delta_i$ to agent $i$ and all agents before $i$: $z_{j,t}:=z_{j,t}+\delta_i\ for\ j\leq i$, update the remaining allocation $B:=B-i\delta_i$ and move on to the next step $i:=i+1$.

If $i\delta_i>B$, we allocate the remaining item fraction $B$ equally between the agent $i$ and agents before $i$, $z_{j,t}:=z_{j,t}+\frac{B}{i}\ for\ j\leq i$ (terminating the algorithm on this round).

If the algorithm gets to $i=N$, we similarly allocate the remaining item fraction equally between everyone $z_{j,t}:=z_{j,t}+\frac{B}{N}\ for\ all\ j\in N$.  

Thus the number of operations on a single round is at most $O(N)$, giving us $O(TN)$ complexity for making the allocations over all rounds.

We now briefly give an argument for correctness. First, as we mentioned before, it is always optimal to allocate the marginal fraction of the item to the agent with the smallest $\frac{v_{i,t}}{\tilde u_{i,t}(z_{i,t})}$ ratio, where $z_{i,t}$ is the allocation given to agent $i$ so far. The second thing we need to note is that, when this ratio is tied between several agents 
\begin{align}
\label{eq:comp_equality}
\frac{v_{i,t}}{\tilde u_{i,t}(z_{i,t})} = \frac{v_{j,t}}{\tilde u_{j,t}(z_{j,t})},
\end{align}

adding equal values $\epsilon$ to their allocation preserves the tie 

$$\frac{v_{i,t}}{\tilde u_{i,t}(z_{i,t}+\epsilon)} = \frac{v_{j,t}}{\tilde u_{j,t}(z_{j,t}+\epsilon)}$$

To see that this is true, recall that $\tilde u_{i,t}(z_{i,t}+\epsilon) = u_{i,t}(z_{i,t})+v_{i,t}\epsilon$, and cancel out equal terms implied by \eqref{eq:comp_equality}.

This implies that the algorithm always allocates the marginal fraction of the item to the agents minimizing the $\frac{v_{i,t}}{\tilde u_{i,t}(z_{i,t})}$ ratio, and thus solves \eqref{eq:comp1} exactly.

\newpage
\bibliographystyle{abbrvnat}
\bibliography{main-bibliography,MARA}



\end{document}